\documentclass[11pt, reqno]{amsart}

\usepackage{amsmath,amsthm,amsfonts,amscd,amssymb}
\usepackage{times,euler,eucal,eufrak,graphicx}
\usepackage{hyperref}
\usepackage{float}
\usepackage{caption}
\usepackage{subcaption}
\usepackage{dirtytalk}

\title[PPT2]
{The PPT square conjecture holds generically for some classes of independent states}
\author {Beno\^{i}t Collins}
\address{Department of Mathematics, Graduate School of Science,
Kyoto University, Kyoto 606-8502, Japan}
\email{collins@math.kyoto-u.ac.jp}
\author {Zhi Yin}
\address{Institute of Advanced Study in Mathematics, Harbin Institute of Technology, Harbin 150006, China}
\email{hustyinzhi@163.com}
\author {Ping Zhong}
\address{Department of Pure Mathematics, University of Waterloo, 
200 University Avenue West, Waterloo, Ontario N2L 3G1, Canada}
\email{ping.zhong@uwaterloo.ca}

\theoremstyle{plain}
\newtheorem{lemma}{Lemma}[section]
\newtheorem{theorem}[lemma]{Theorem}
\newtheorem{proposition}[lemma]{Proposition}
\newtheorem{corollary}[lemma]{Corollary}

\theoremstyle{definition}

\theoremstyle{remark}
\newtheorem*{remark}{Remark}

\newcommand{\un}{1\mkern -4mu{\rm l}}

\def\>{\rangle}
\def\<{\langle}

\def\be{\begin{equation}}
\def\ee{\end{equation}}

\begin{document}

\begin{abstract}
Let $|\psi\rangle\langle \psi|$ be a random pure state on $\mathbb{C}^{d^2}\otimes \mathbb{C}^s$, where $\psi$ is a random unit vector uniformly distributed on the sphere in $\mathbb{C}^{d^2}\otimes \mathbb{C}^s$.
Let $\rho_1$ be random induced states $\rho_1=Tr_{\mathbb{C}^s}(|\psi\rangle\langle\psi |)$ whose distribution is $\mu_{d^2,s}$; and let $\rho_2$ be random induced states following the same distribution $\mu_{d^2,s}$ independent from $\rho_1$.
Let $\rho$ be a random state induced by the entanglement swapping of 
$\rho_1$ and $\rho_2$.
We show that the empirical spectrum of 
$\rho- \un/d^2$ 
converges almost surely to the Marcenko-Pastur law with parameter $c^2$
as $d\rightarrow \infty$ and  $s/d \rightarrow c$.
 As an application, we prove that the state $\rho$ is separable 
 generically if $\rho_1, \rho_2$ are PPT entangled.

\end{abstract}

\maketitle

\section{Introduction}

Quantum entanglement \cite{HHHH2007} is a resource which can be widely used in quantum information process. 
Assuming that Alice and Bob share an entangled state,
the quantum teleportation claims that 
Alice can teleport any unknown quantum state to Bob by using a classical channel, 
although it is known that it is impossible to create an ideal copy of an arbitrary 
quantum state by no-cloning theorem.
Entanglement swapping can create entanglement between systems which never interact,
and it is also the teleportation of entanglement of the maximally entangled states 
which is an  important technique for teleportation over long distances \cite{ZZHE1993}.

While the maximally entangled state plays an important role in typical teleportation protocols, there exists another 
important type of entangled states, called the bound entangled states \cite{HHH1998}, 
from which one can not distill maximally entangled states
under local operations and classical communication.
For example, states satisfying the positive partial transpose (PPT) property are bound entangled states \cite{HHH1998}.
Unfortunately, in any deterministic teleportation protocol, the performance of the bound entangled state is not better than the separable (classical) state as shown in \cite{HHH1998}.
In view of this, it is natural to wonder what happens to the entanglement swapping protocol with such states.

The so called PPT square conjecture first appeared in B\"{a}uml's thesis in the following \emph{state} form \cite{B2010}:
{\it assume that Alice and Charlie share bound entangled state and that Bob and Charlie also share bound entangled state;
then the state of Alice and Bob, conditioned on any measurement by Charlie, is separable.}
In other words, this conjecture suggests that 
the state obtained by entanglement swapping protocol of PPT entangled states is separable. 

Up to now, there are some evidences to support the PPT square conjecture \cite{BCHW2015, B2010}. 
For example, 
Hermes announced 
that this conjecture is true for the states on $\mathbb{C}^3 \otimes \mathbb{C}^3$ \cite{MM2017} recently. However, one main difficulty to study this conjecture is that we can not describe the set of all bound entangled states and the conjecture remains open.

In this article, we use two methods widely used in quantum information theory, random matrx theory (RMT) and asymptotic geometric analysis (AGA), to study
the PPT square conjecture. 
The RMT was heavily used in 
the non-additivity problem of quantum channels \cite{BH2010, CN2010, CN2011, CNY2011, FN2015, H2009, HW2008}, 
and AGA was used to estimate the geometric volume of quantum states with different properties \cite{A2012, ASY2014, CNY2011, SZ04, ZS01}. 

We outline our approach as follows. We consider two induced states on $\mathbb{C}^d \otimes \mathbb{C}^d$
chosen randomly with distribution $\mu_{d^2, s}$, where $s$ is the dimension of the environment. 
By Aubrun \cite{A2012} and Aubrun-Szerek-Ye's work \cite{ASY2014}, 
we can choose the parameters $s$ and $d$ properly, such that
the induced states chosen are PPT entangled generically.
In other words, the states are PPT entangled with high probability
as $d \rightarrow \infty$ and $s/d \rightarrow \infty$. 
Then we will study the separability of the state which is induced by the entanglement swapping protocol of the states. This is 
done when
these two states are chosen independently.
Our work is divided
 into two parts. 
 Firstly, we consider the model of entanglement swapping protocol of two Wishart matrices
 and then calculate the moments of the rescaled radon matrix model. 
 By using tools of RMT, 
we are able to obtain the large $N$ limits of this model. Secondly, by using AGA and our limit theorem, we show that the state is separable
with high probability as $d \rightarrow \infty$ and $s/d \rightarrow \infty$. 
In this sense, we prove that the PPT square conjecture holds generically if the states are chosen independently. 
Moreover, we also consider the random model where the two states are chosen to be the same. 
However, we can only get a weak version of limit theorem, hence 
we are not able 
to use AGA to describe the 
separability of the induced state directly.

The paper is organized as follows. After this introduction, we collect some 
relevant results from RMT in Section 2. We then introduce our random matrix models and 
prove some limit theorems in Section 3. 
We apply our results to PPT square conjecture in Section 4 
and end the paper with conclusions and some questions.

\section{Preliminaries}
In this section, we review some relevant results in combinatorics and graphic Gaussian calculus. 
\subsection{Some combinatorial facts}

Let $I$ be a linearly ordered set of $p$ elements and identify it with 
the interval $[p] = \{1, 2, \ldots, p\}$.
Denote by $S_I$ the set of permutations of elements in $I$. 
For convenience, we also denote by $S_p$ the set of of permutations of elements in $[p]$. 
Given a permutation $\sigma\in S_I$,
we denote by $|\sigma|$ the minimal number of transpositions that multiply to $\sigma$ and
by $\#\sigma$ the number of cycles of $\sigma$.
We have the the following equation
\begin{equation}\label{eq:permutation}
\#\sigma = |I|- |\sigma|.
\end{equation}
Let $d(\sigma,\tau)=|\sigma^{-1}\tau|$,
then we have
following triangle inequality:
\begin{equation}\label{eq:geo}
|\sigma^{-1} \tau| + |\tau^{-1} \pi| \geq |\sigma^{-1} \pi|.
\end{equation}
Hence it defines a distance on $S_I$. We also call $|\sigma|$ the length of $\sigma$.
If the equality in (\ref{eq:geo}) holds, 
we say $\sigma, \tau, \pi$ satisfy the geodesic condition
and denote it as $\sigma-\tau-\pi$. 
We denote the permutations from $S_p$ which lie on a geodesic from $id$ to 
the full cycle $\gamma:=(1,2,\cdots, p)$ by
\[
\begin{split}
S_{NC}(\gamma):&=\{ \pi\in S_p| |\pi|+|\pi^{-1}\gamma|=p-1 \}\\
   &=\{ \pi \in S_p: id-\pi -\gamma\}.
\end{split}
\]
For $\sigma, \pi\in S_{NC}(\gamma)$, we say that $\sigma\leq \pi$ if $\sigma$ and $\pi$ lie on the same geodesic and $\sigma$ comes before $\pi$. The set $S_{NC}(\gamma)$
endowed with $\say{\leq}$ becomes a poset. 
We refer the reader to \cite{NS2006}
for more details.

We call $\pi =\{V_1, \ldots, V_r\} $ a partition of the set $[p]$ if the sets $V_i$ ($i=1, \ldots, r$) are pairwise disjoint, non-empty subsets of $[p]$ such that $V_1 \cup \cdots \cup V_r =[p].$ We use $\#\pi$ to denote the number of blocks of $\pi.$ Given two elements $a, b \in [p]$, 
we write $a \sim_\pi b$ if $a$ and $b$ belong to the same block of $\pi.$
A partition $\pi$ is called crossing if there exist $a_1< b_1 < a_2 < b_2 \in [p]$ such that $a_1 \sim_\pi a_2 \not\sim_\pi b_1 \sim_\pi b_2.$ We call $\pi$ \emph{non-crossing} partition if $\pi$ is not crossing. Denote
by $NC(p)$ the set of all non-crossing partitions of $[p]$. 
We also denote by $NC(I)$ the set of all non-crossing partitions of 
the linearly ordered set $I$. 

A partition can be naturally identified with a permutation. 
We will use the following identification of non-crossing partitions due to Biane \cite{B1997} (see also \cite[Lecture 23]{NS2006}). 

\begin{lemma}\label{lem: 0}
Let $\gamma = (1, 2, \ldots, p)$ be the full cycle of $S_p$. 
There is a bijection between $NC(p)$ and the set $S_{NC}(\gamma)$ which preserves 
the poset structure. 
\end{lemma}

We end this subsection by the following two technical lemmas. 

\begin{lemma}\label{lem: 1}
Denote by $NC^0(p) = \{\pi \in NC(p): \pi \; \text{has no singletons}\}$.
We then have 
\begin{equation}\label{eq: lem1}
\begin{split}
\sum_{I \subset [p]} (-1)^{p-|I|}  \sum_{ \pi \in NC (I)} c^{p-2|I|+2(\# \pi)}
 =  \sum_{\pi \in NC^0(p)} c^{2(\# \pi)-p}.
\end{split}
\end{equation}
\end{lemma}

\begin{proof}
It suffices to prove that only the terms with $I = [p]$ and $\pi \in NC[p]$ without singletons in the above sum survive. 
Given a subset $J\subset I$ and a non-crossing partition $\pi' \in NC^0(J)$ having no
singletons, we denote by $[\pi']$ the set of non-crossing partitions in $NC(I)$ which extend $\pi'$ by adding singletons:
$$[\pi'] = \{\pi \in NC(I),   \pi = \pi' \cup \big\{ \text{singletons in}\;  I \backslash J \} \big\}.$$
 Noting the fact that for $\pi \in [\pi']$, we have $\# \pi = \# \pi' + |I| -|J|$ and
 every non-crossing partition $\pi\in NC(I)$ except $\{ \{1\}, \{2\}, \cdots, \{p\} \}$
 can be decomposed as $\pi=\pi' \cup \{ \text{singletons}\}$ by removing singletons from $\pi$. We have
\begin{equation*}
\begin{split}
\text{LHS of \eqref{eq: lem1}} &= \sum_{J \subset [p]}  \sum_{J \subset I \subset [p]} (-1)^{p-|I|}
 \sum_{\substack{\pi \in NC(I) \cap [\pi']\\ \pi'\in NC^0(J)}} c^{p-2|I| +2 \# \pi}\\
& =  \sum_{J \subset [p]} \sum_{\pi' \in NC^0(J)} c^{p-2|J| +2 \# \pi'} \sum_{J \subset I \subset [p]} (-1)^{p-|I|}\\
& =  \sum_{J \subset [p]} \sum_{\pi' \in NC^0(J)} c^{p-2|J| +2 \# \pi'} \delta_{|J|, p} = \text{ RHS of \eqref{eq: lem1}},
\end{split}
\end{equation*}
where we have used the fact that $\sum_{J \subset I \subset [p]} (-1)^{p-|I|} =0$
if $|J|\neq p$.
\end{proof}

\begin{lemma}\label{lem: 2}
Let $\gamma = (1, 2, \ldots, p)$ be the full cycle of $S_p.$ Let $\delta = \gamma \oplus \gamma \in S_{2p}$ and $\beta \in S_{2p}$ such that $\beta(i) = \beta(i+p), i=1, \ldots, p.$ Then for any  $\pi \in 
\{\pi \in S_{2p}: id-\pi -\delta\},$ $\pi$ can be decomposed into $\pi_1 \oplus \pi_2,$ where $id-\pi_1, \pi_2 - \gamma.$ Moreover, we have   
$$\# (\beta^{-1} \pi) = \# (\pi_1 \pi_2).$$
\end{lemma}

\begin{proof}
We shall use the identification between non-crossing partitions and geodesic permutations (see Lemma \ref{lem: 0}). 
The non-crossing partition $\pi\leq \{ \{1,2,\ldots, p \}, \{p+1,p+2,\ldots, 2p \} \}\subset NC(2p)$, thus 
$\pi$ can be decomposed into $\pi=\pi_1\oplus\pi_2 $, where $\pi_1\in NC(p)$ and $\pi_2$ is a non-crossing partition 
of the set $\{ p+1, p+2,\ldots, 2p\}$.

Let $V$ be a cycle of $\beta^{-1}\pi= \beta^{-1}( \pi_1\oplus\pi_2)$.
Since $\beta(\{1,2,\ldots, p\})=\{ p+1, p+2,\ldots, 2p\}$
and $ \{ p+1, p+2,\ldots, 2p\} $ is invariant under $(\beta^{-1}\pi)^2$, we see $V$ must intersect with $\{p+1, p+2,\ldots, 2p \}$.
Now let $1\leq a\leq p$, we have $\beta^{-1}( \pi_1\oplus\pi_2)(p+a)=\beta^{-1} (p+\pi_2(a))=\pi_2(a)$
and $ \beta^{-1}( \pi_1\oplus\pi_2)(\pi_2(a))=\beta^{-1} ( \pi_1(\pi_2(a)))=p+\pi_1(\pi_2(a)) $.
Hence, by identifying $\{p+1, p+2,\ldots, 2p\}$ with $\{1,2,\ldots, p \}$,
the action of $ (\beta^{-1}\pi)^2$ is the same as the action of $\pi_1\pi_2$. In particular, we have
$ \#(\beta^{-1}\pi)=\#(\pi_1\pi_2)$.
\end{proof}

\subsection{Graphical Gaussian calculus}

In \cite{CN2010, CN2011},  
The first named author and Nechita introduced a graphical formulation of the Weingarten calculus, 
which is very useful to evaluate moments of the output of the quantum channel we are interested in. 
Let us briefly review the main ideas and refer the reader to the original article for details.

A diagram is a collection of boxes with certain decorations and possibly wires, which connect
the boxes along their decorations according to some rules. Each decoration can be either filled (black) or empty (white), 
which corresponds to vector spaces or their dual spaces. And each wire connecting shapes attached to boxes corresponds 
to a tensor of a vector space with its dual, which produces a partial trace operation.
A diagram consisting of boxes and wires corresponds is denoted by $\mathcal{D}$.

In Figure \ref{fig1}, we give some simple examples of diagrams.
In Figure \ref{fig1} (a), the matrix $G \in M_{nk}$ is represented as a box with two decorations,
the round one stands for the $n$-dimensional Hilbert space and the square one stands for the $k$-dimensional part.
The wire connecting the round decorations stands for tracing over the $n$-dimensional part,
where we identify $\mathbb{C}^{nk}\cong \mathbb{C}^n\otimes \mathbb{C}^k$.

After taking partial trace, Figure \ref{fig1} (a) ends with a matrix in $M_k$.
The diagram in Figure \ref{fig1} (b) represents
the (non-normalized) maximally entangled state
and Figure \ref{fig1} (c) shows
the equivalence of two diagrams corresponding to $G^*$ and $\overline{G}$ respectively.
  \begin{figure}
  \centering
  \includegraphics[height=0.5\textwidth]{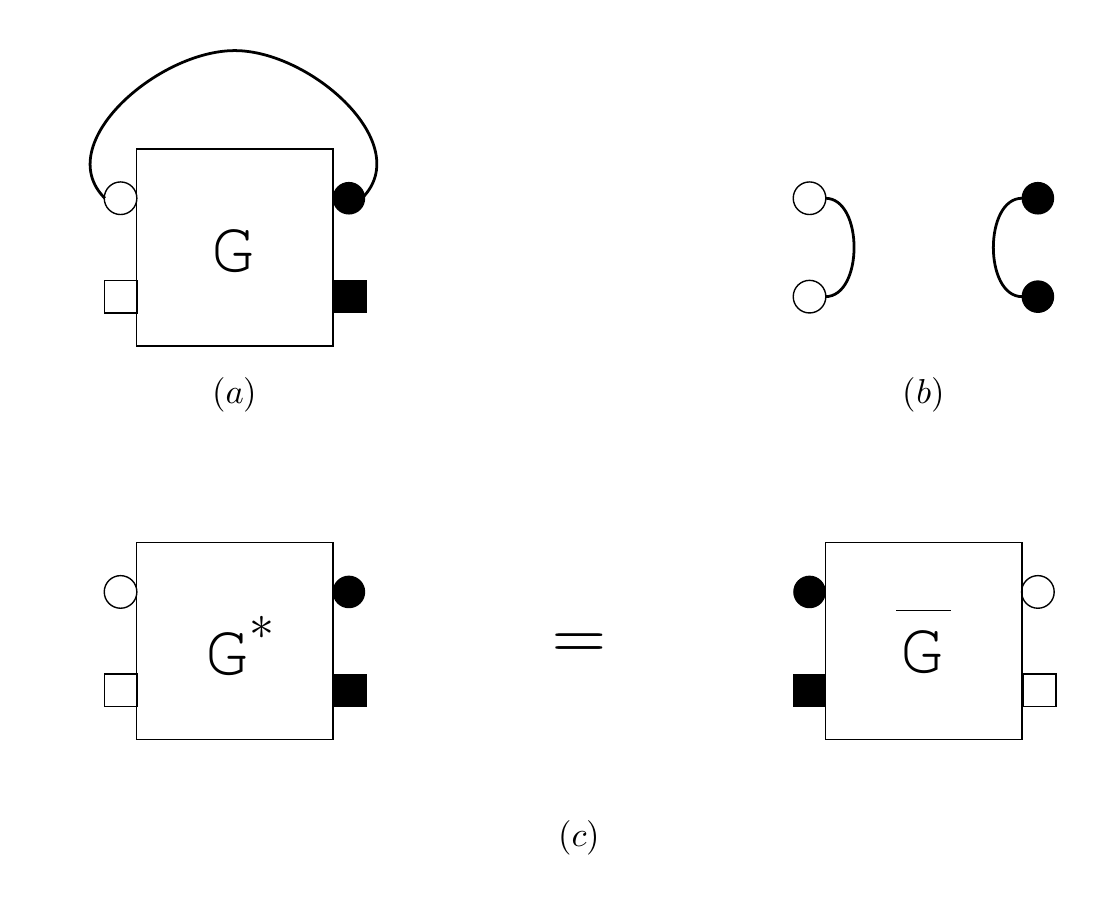}
  \caption{diagram (a) for $Tr_n (G)$, (b) for the non-normalized maximally entangled state, and (c) for the identification between $G^*$
  and $\overline{G}$.}
  \label{fig1}
  \end{figure}

Let us now describe very briefly how to compute expectation values of diagrams containing boxes of $G$ and $\overline{G}$,
where $G$ is a matrix whose entries are independent random variables with the same distribution $N(0, 1)$. 
We fill each box a matrix $G$, or its relative $ \overline{G}$.
To emphasize that $G$ is involved in the diagram, we use $\mathcal{D}(G)$ to represent such diagrams.
Given a permutation $\pi$, a wire will be added to connect a decoration labeling white (resp. black) of
the box $G$ having index $i$, with the same decoration labeling black (resp. white) of the box $\overline{G}$
having index $\pi(i)$.
We must add enough wires so that every decoration labeling white must
be paired with the same decoration labeling black.

The operation gives us a new diagram $\mathcal{D}(G)_{\pi}$, which is called a removal of the original diagram.
Using this operation and the Wick formula, one can describe the expectation $\mathbb{E}(\mathcal{D})$ of diagrams $\mathcal{D}$ as the following equation \cite[Theorem 3.2]{CN2012} formally:
\begin{equation*}
\mathbb{E}(\mathcal{D}(G))=\sum_{\pi}\mathcal{D}(G)_{\pi}.
\end{equation*}

In order to calculate $\mathcal{D}(G)_{\pi},$ we have to count the contributions for each $\pi$ of every decoration.
For instance, suppose there are two kinds of decorations, say
$\includegraphics[scale=0.14]{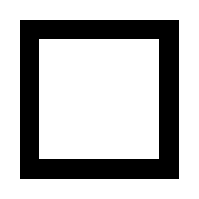}$
and
$\includegraphics[scale=0.12]{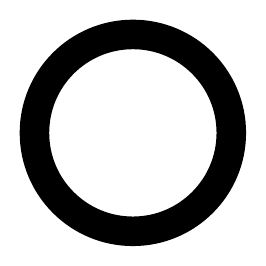}$
and their corresponding dimensions are $n$ and $k$ respectively. Thus for each $\pi,$
\begin{equation*}
\mathcal{D}(G)_{\pi} = n^{\#\includegraphics[scale=0.1]{square_w.pdf}}
k^{\# \includegraphics[scale=0.08]{circle_w.pdf}},
\end{equation*}
where $\# \includegraphics[scale=0.15]{square_w.pdf}$
(resp. $\# \includegraphics[scale=0.12]{circle_w.pdf}$)
denotes the number of the loops which connects the $\includegraphics[scale=0.15]{square_w.pdf}$ (resp.
$\includegraphics[scale=0.12]{circle_w.pdf}$) decorations.

In this paper, our random matrix models are related to the Wishart matrix $W \in W(n,s)$,
which is an $n \times n$ random matrix of the form $W= GG^\ast,$ where $G$ is a $n \times s$ 
random matrix whose entries are independent identically distributed random variables with the standard complex 
Gaussian distribution.
We will calculate the moments of certain random matrix models
using the graphical calculus.

\section{Entanglement swapping process of two Wishart matrices}
Let $W_1, W_2 \in W(d_1d_2, s)$ be two Wishart matrices with the same parameters $(d_1d_2, s)$, we would like to analyze the spectrum distribution of $d_2^2 \times d_2^2$ matrix $W$ which is obtained by the following 
\emph{entanglement swapping}
process of 
$W_1$ and $W_2:$ 
\begin{equation}\label{eqn:004}
W = \frac{1}{d_1}Tr_{d_1} \left[( W_1 \otimes W_2) \; P_{d_1} \right],
\end{equation}
where $Tr_{d_1}$ is the partial trace and $P_{d_1}$ is the Bell projection on the 
two parties with dimension $d_1$.
More precisely, 
let $H_i=\mathbb{C}^{d_i}\otimes\mathbb{C}^{d_i}$ and
we identify a $d_1^2d_2^2 \times d_1^2d_2^2$ matrix 
with an operator on $B(H_1) \otimes B(H_2)$. 
Hence for any 
$T =\sum_{i,j=1}^{d_1^2} E_{ij}\otimes T_{ij} \in B(H_1) \otimes B(H_2),$ 
$Tr_{d_1} (T) = \sum_{j=1}^{d_1^2} T_{jj}$,
where $E_{ij}$ denotes the elementary matrix with $1$ at the $(i,j)$-entry and $0$ at other entries. Using Dirac notation, the Bell vector is $|\phi\rangle=\frac{1}{\sqrt{d_1}}\left(|1\rangle \otimes|1\rangle +\cdots |d_1\rangle\otimes |d_1\rangle\right)$, and the Bell projection is
$P_{d_1} =d_1|\phi\rangle\langle\phi|= \sum_{i,j=1}^{d_1} |i\rangle\langle j|\otimes |i\rangle\langle j|$,
where $\{ e_1, \cdots, e_{d_1}\}$ is an orthonormal basis of $\mathbb{C}^{d_1}$.

\subsection{Moment formulas of the random matrix $W$}

\begin{proposition}\label{prop:Moments}
The moments of $W$ are given by the following formula.
\begin{enumerate}
\item  Case I: If $W_1$ and $W_2$ are chosen independently,
\begin{equation}\label{eq: Moments-indep1}
\mathbb{E} Tr [W^p] = \sum_{\pi_1, \pi_2 \in S_p} d_2^{\# (\gamma^{-1} \pi_1) + \# (\gamma^{-1} \pi_2)}  s^{\# (\pi_1) + \# (\pi_2)} d_1^{\#(\pi_1^{-1} \pi_2)-p},
\end{equation}
where $\gamma  = (1, 2, \ldots, p)$ is the full cycle of $S_p.$
\\
\item Case II: If $W_1=W_2,$ 
\begin{equation}\label{eq: Moments-equal1}
\mathbb{E} Tr [W^p] = \sum_{\pi \in S_{2p}} d_2^{\# (\delta^{-1} \pi)}  s^{\# (\pi)} d_1^{\#(\beta^{-1} \pi)-p},
\end{equation}
where $\delta = \gamma \oplus \gamma \in S_{2p} $,
$\gamma=(1,2,\cdots, p)$
 and $\beta \in S_{2p}$ is defined as $\beta(i) = \beta (i+p).$
\end{enumerate}
\end{proposition}

\begin{proof}

The diagrams corresponding to $W$ are given in Figure \ref{fig 2}. 

  \begin{figure}[H]
  \centering
  \includegraphics[height=0.4\textwidth]{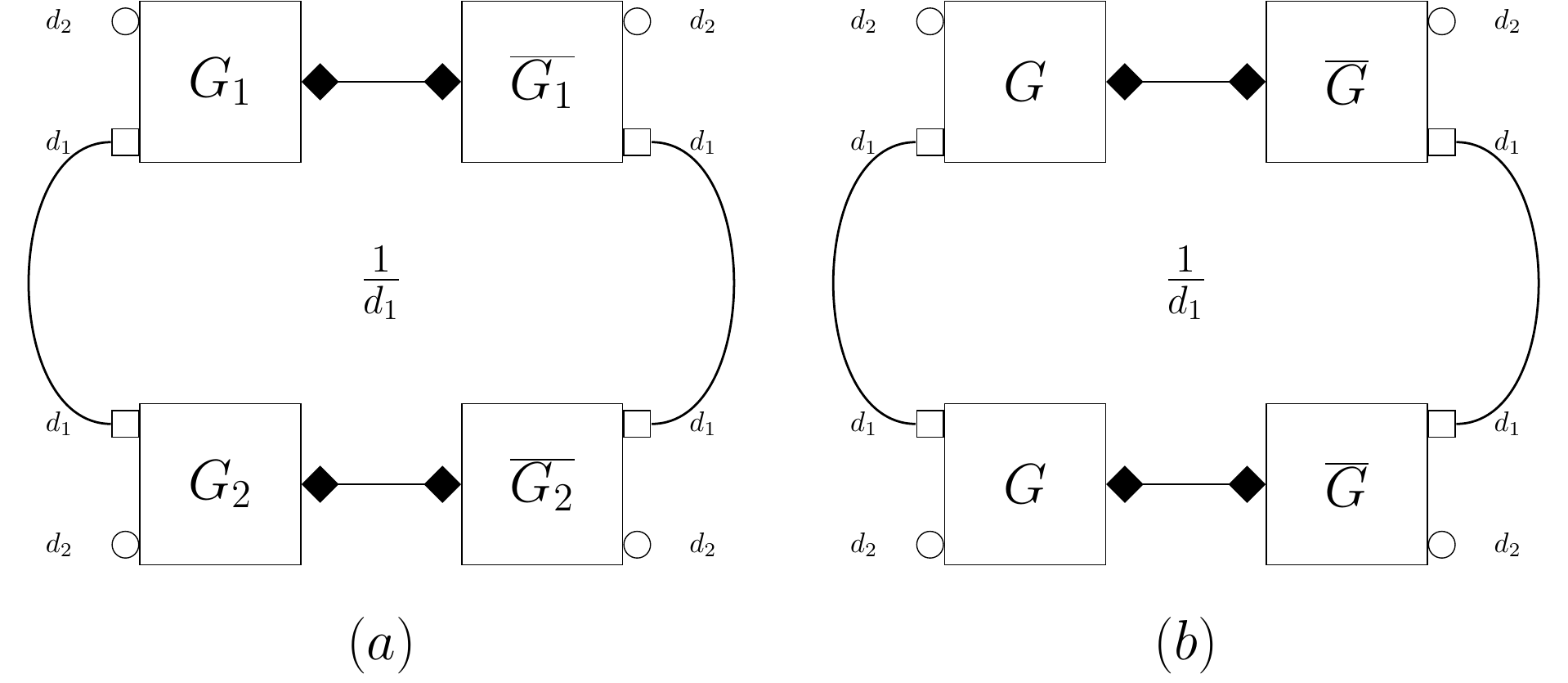}
  \caption{diagram (a) for the case I, (b) for the case II.}
  \label{fig 2}
  \end{figure}

For the case I, by using the graphic Gaussian calculus, we have the following formula:
\begin{equation*}
\mathbb{E} Tr [W^p]  =\sum_{\pi_1,\pi_2 \in S_{p}}  \mathcal{{D}}(G_1, G_2)_{\pi_1, \pi_2}.
\end{equation*}

To obtain the $\mathcal{{D}}(G_1, G_2)_{\pi_1, \pi_2}$. We label the $G_1$ (resp. $G_2$)
and the $\overline{G_1}$ (resp. $\overline{G_2}$) boxes with $1, \ldots, p.$  A removal $(\pi_1, \pi_2) \in S_{p} \times S_{p}$ of
the boxes $G_1, G_2$ and $\overline{G_1}, \overline{G_2}$ connects the decorations in the following way:
\begin{enumerate}
\item the white (resp. black) decorations of the $i$-th $G_1$ block are paired with the white (resp. black) decorations of the $\pi_1 (i)$-th $\overline{G_1}$ block;
\item the white (resp. black) decorations of the $i$-th $G_2$ block are paired with the white (resp. black) decorations of the $\pi_2 (i)$-th $\overline{G_2}$ block.
\end{enumerate}

We can now compute the contributions for each pairing $(\pi_1, \pi_2):$
\begin{enumerate}\item white "$\includegraphics[scale=0.12]{circle_w.pdf}$"-loops: $d_2^{\# (\gamma^{-1} \pi_1) + \# (\gamma^{-1} \pi_2)} $;
\item white "$\includegraphics[scale=0.15]{square_w.pdf}$"-loops: $d_1^{\#(\pi_1^{-1} \pi_2)};$
\item black "$\includegraphics[scale=0.13]{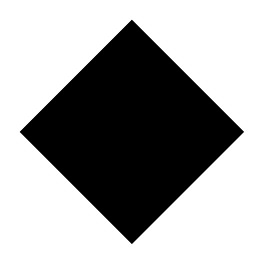}$"-loops: $s^{\# \pi_1 + \# \pi_2};$
\item normalization factors ${d_1}^{-1}$ from the $P_{d_1}$: ${d_1}^{-p}.$
\end{enumerate}

Hence
\begin{equation*}
\mathcal{{D}}(G_1, G_2)_{\pi_1, \pi_2} = d_2^{\# (\gamma^{-1} \pi_1) + \# (\gamma^{-1} \pi_2)}  s^{\# (\pi_1) + \# (\pi_2)} d_1^{\#(\pi_1^{-1} \pi_2)-p},
\end{equation*}
which completes the proof of the case I. 
\vspace{3mm}

For Case II, we label the $G$ and the $\overline{G}$ boxes in the following manner: $1^T, \ldots, p^T$ for the $G$ (resp. $\overline{G}$)
boxes that are on the top of the diagram and $1^B, \ldots, p^B$ for the $G$ (resp. $\overline{G}$) boxes that are on the 
bottom of the diagram. We shall resign the labels as $\{1^T, \ldots, p^T, 1^B, \ldots, p^B\}\simeq \{1, \ldots, 2p\}.$
With this notation, the two fixed permutations $\delta$ and $\beta \in S_{2p}$ introduced in the main text are the following: 
for all $i$,
\begin{equation*}
\delta(i^T) = (i+1)^T, \delta(i^B) =(i+1)^B, \;\; \text{and} \;\; \beta(i^T) = i^B, \beta(i^B) = i^T.
\end{equation*}

Now we have the following formula:
\begin{equation*}
\mathbb{E} Tr [W^p]  =\sum_{\pi \in S_{2p}}  \mathcal{{D}}(G)_{\pi}.
\end{equation*}

A removal $\pi \in S_{2p}$ of
the boxes $G$ and $\overline{G}$ connects the decorations in the following way: the white (resp. black) decorations of the $i$-th $G$ block are paired with the white (resp. black) decorations of the $\pi (i)$-th $\overline{G}$ block.

On the other hand,  the contributions for each pairing $\pi$ can be computed as the following:
\begin{enumerate}\item white "$\includegraphics[scale=0.12]{circle_w.pdf}$"-loops: $d_2^{\# (\delta^{-1} \pi)} $;
\item white "$\includegraphics[scale=0.15]{square_w.pdf}$"-loops: $d_1^{\#(\beta^{-1} \pi)};$
\item black "$\includegraphics[scale=0.13]{diamond_b.pdf}$"-loops: $s^{\# \pi};$
\item normalization factors ${d_1}^{-1}$ from the $P_{d_1}$: ${d_1}^{-p}.$
\end{enumerate}

Hence
\begin{equation*}
\mathcal{{D}}(G)_{\pi} = d_2^{\# (\delta^{-1} \pi)}  s^{\# \pi} d_1^{\#(\beta^{-1} \pi)-p},
\end{equation*}
which finishes the proof.
\end{proof}


\subsection{Moments of rescaled matrix of $W$ in the asymptotic regime $d_1, d_2 \rightarrow \infty, s/d_2 \rightarrow c$}

We introduce the following rescaled matrix:
\begin{equation}\label{eqn:007}
Z = d_2s \left(   \frac{W}{d_2^2 s^2}  - \frac{\un}{d_2^2}  \right).
\end{equation}

Then we have the following theorem:
\begin{theorem}
Let $c>0$ be a constant, then in case I and II, the moments of $Z$ under the asymptotic regime ($d_1, d_2 \rightarrow \infty$ and $s/d_2 \rightarrow c$) are given by
\begin{equation}\label{eqn:008}
\lim_{\begin{subarray}{c}\\ d_1, d_2 \rightarrow \infty \\ s/d_2 \rightarrow c \end{subarray}} \frac{1}{d_2^2}\mathbb{E}  Tr [Z^p] =  \sum_{\pi \in NC^0(p)} c^{2(\# \pi)-p}.
 \end{equation}
\end{theorem}

\begin{proof}
By binomial identity, in both cases we have 
$$m_p (Z):= \frac{1}{d_2^2}  \mathbb{E}Tr [Z^p] = \frac{1}{d_2^2} \sum_{I \subset [p]}  \left( -\frac{s}{d_2} \right)^{|I^c|} \left(\frac{1}{d_2 s} \right)^{|I|}Tr [W^{|I|}].$$
Case I: Let $s = c d_2,$ we have  
 \begin{equation*}
\begin{split}
 m_p (Z)  
& =  \frac{1}{d_2^2}  \sum_{I \subset [p]} \left( -\frac{s}{d_2} \right)^{|I^c|} \left(\frac{1}{d_2 s} \right)^{|I|} \cdot \\
& \cdot \sum_{\pi_1, \pi_2 \in S_I} d_2^{ \#(\gamma_I^{-1} \pi_1) + \#(\gamma_I^{-1} \pi_1)} s^{\#(\pi_1)+\#(\pi_2)} d_1^{\#(\pi_1^{-1} \pi_2)-|I|}\\
& =  \sum_{I \subset [p] }  (-1)^{p-|I|} \sum_{\pi_1, \pi_2 \in S_I} d_2^{f_I(\pi_1, \pi_2)} c^{ -|\pi_1|-|\pi_2| +p} d_1^{\#(\pi_1^{-1} \pi_2)-|I|},
\end{split}
\end{equation*}
where $f_I(\pi_1, \pi_2) = 2|I| -|\gamma_I^{-1} \pi_1| -|\gamma_I^{-1} \pi_2|-|\pi_1|-|\pi_2|-2.$ Note that the power of $d_2$ becomes
\begin{equation*}
\begin{split}
f_I(\pi_1, \pi_2)  \leq 2 ( |I|-1-|\gamma_I|)=0,
\end{split}
\end{equation*}
where we used the fact that $|\gamma_I| = |I|-1$. Therefore the power of $d_2$ terms converge to zero as 
$d_2 \rightarrow \infty$, except the case when satisfy the following geodesic condition $id - \pi_1, \pi_2 - \gamma_I.$ 
Moreover,   for $\pi_1, \pi_2 \in S_I,$ $ \#(\pi_1^{-1} \pi_2) \leq |I|$ and the equality holds if and only if $\pi_1 = \pi_2$. 
Hence when $d_1, d_2 \rightarrow \infty, s/d_2 \rightarrow c$ the only terms in the moments $m_p(Z)$ which survive 
are those for which $id-\pi_1= \pi_2-\gamma_I$ and $|\pi_1|+ |\pi_2| = p$. That is, 
\begin{equation}\label{eq: moments-limit}
\begin{split}
\lim_{\begin{subarray}{c}\\ d_1, d_2 \rightarrow \infty\\ s/d_2 \rightarrow c \end{subarray}} m_p (Z) &=  \sum_{I \subset [p] }  (-1)^{p-|I|} \sum_{ id-\pi -\gamma_I}  c^{p -2|\pi|}\\
&  = \sum_{I \subset [p] }  (-1)^{p-|I|} \sum_{ \pi \in NC (I)} c^{2\# \pi-2|I| +p}\\
& = \sum_{\pi \in NC^0(p)} c^{2(\# \pi)-p},
\end{split}
\end{equation}
where we used Lemma \ref{lem: 1}. This completes the proof for case I.
\\

\noindent Case II: $W_1=W_2$, by applying Proposition 3.1, we have 
\begin{equation}\label{eq: almost}
\begin{split}
m_p(Z) & =  \sum_{I \subset [p]} \frac{1}{d_2^2}  \left( -\frac{s}{d_2} \right)^{|I^c|} \left(\frac{1}{d_2 s} \right)^{|I|} \sum_{\pi \in S_{I \cup I}}  d_2^{\# (\delta_I^{-1} \pi)}  s^{\# (\pi)} d_1^{\# (\beta_I^{-1} \pi)}\\
& = \sum_{I \subset [p]} (-1)^{p-|I|} \sum_{\pi \in S_{I \cup I} }d_2^{2|I| -|\delta_I^{-1} \pi| -|\pi| -2} c^{-|\pi|+p} d_1^{\# (\beta_I^{-1} \pi)-|I|}.
\end{split}
\end{equation}
We set  $s =c d_2$. 
The power of $d_2$ is given by
\begin{equation*}
\begin{split}
2|I| -|\delta_I^{-1} \pi| -|\pi| -2 \leq 2|I| -|\delta_I| -2 =0,
\end{split}
\end{equation*}
where we have used the fact that $|\delta_I| = 2|I| -2.$ Therefore, when $d_2 \rightarrow \infty, s/d_2 \rightarrow c$, 
the only terms in the moments $m_p(Z)$ which survive are those for which $id-\pi-\delta_I$, and we have 
\begin{equation*}
\begin{split}
\lim_{\begin{subarray}{c}\\ d_2 \rightarrow \infty \\ s/d_2 \rightarrow c \end{subarray}} m_p (Z) &=  \sum_{I \subset [p] }  (-1)^{p-|I|} \sum_{ id-\pi -\delta_I } c^{p-|\pi|} d_1^{\# (\beta_I^{-1} \pi)-|I|}.
\end{split}
\end{equation*}
By Lemma \ref{lem: 2}, $\pi$ can be decomposed into direct sum of $\pi_1 \oplus \pi_2$ such that $id- \pi_1, \pi_2 -\gamma_I,$ and $\# (\beta_I^{-1} \pi) = \#(\pi_1 \pi_2).$ Hence we have 
\begin{equation*}
\begin{split}
\lim_{\begin{subarray}{c}\\ d_1, d_2 \rightarrow \infty \\ s/d_2 \rightarrow c \end{subarray}} m_p (Z) &=  \sum_{I \subset [p] }  (-1)^{p-|I|} \sum_{id-\pi_1 = \pi_2^{-1}-\gamma_I} c^{p-|\pi_1|-|\pi_2|}\\
&= \sum_{I \subset [p] }  (-1)^{p-|I|} \sum_{id-\pi_1 = \pi_2-\gamma_I} c^{p-|\pi_1|-|\pi_2|}\\
& = \sum_{I \subset [p] }  (-1)^{p-|I|} \sum_{ id-\pi -\gamma_I } c^{p- 2|\pi|}\\
&  = \sum_{I \subset [p] }  (-1)^{p-|I|} \sum_{ \pi \in NC (I)} c^{2\# \pi-2|I| +p}\\
& = \sum_{\pi \in NC^0(p)} c^{2(\# \pi)-p},
\end{split}
\end{equation*}
where we used the identification between the set  
$\{ \pi \in S_I: id-\pi-\gamma\}$ and $NC(I)$. 
\end{proof}

Recall that the 
density of the Marcenko-Pastur distribution with parameter $c$ is given by
\begin{equation*}
\mu_{MP} (x) = \frac{\sqrt{4c -(x-1-c)^2}}{2\pi x} \un_{[(\sqrt{c}-1)^2, (\sqrt{c}+1)^2]} (x) dx. 
\end{equation*}
The density function of the standard semicircular distribution  is 
\begin{equation*}
\mu_{SC} (x) = \frac{1}{2\pi} \sqrt{4-x^2} \un_{[-2, 2]} (x) dx.
\end{equation*}

We have the following result similar to  \cite[Corollary 2.4, 2.5]{CNY2011}.
For completeness, we provide the proof. 
\begin{corollary}\label{cor: 1}
As $d_1, d_2 \rightarrow \infty, s/d_2 \rightarrow c,$ the random matrix $Z$ converges in moments to a centered 
Marchenko-Pastur distribution of parameter $c^2$
(rescaled by $c$). 

If $d_1, d_2 \rightarrow \infty, s/d_2 \rightarrow \infty,$ $Z$ converges in moments to the semicircular distribution. 
\end{corollary}
\begin{proof}
It is known that all the free cumulants of a centered Marchenko-Pastur (free Poisson) distribution with paramemter $c^2$ are equal to $c^2$, except the first one, which is zero (see \cite[Lecture 12]{NS2006}). Hence, we read from (\ref{eqn:008}) and the moment-free cumulant formula that the distribution determined by the moment series (\ref{eqn:008}) is 
the centered 
Marchenko-Pastur distribution of parameter $c^2$ by recalling the factor $c$. 

If $d_1, d_2 \rightarrow \infty, s/d_2 \rightarrow \infty$, only when $p$ is even and $\#\pi =p/2$, the moment is nonzero. 
In this case, the set $\{\pi \in NC^0(2n)| \#\pi=n \}$ is the same as 
the noncrossing pair partitions of $[2n]$, whose cardinality is the $n$-th Catalan number. Hence, the limit distribution is the semicircular distribution. 
\end{proof}


\subsection{Almost surely convergence of $Z$ in the asymptotic regime 
$d_1= d_2 \rightarrow \infty, s/d_2 \rightarrow c$
}

$\ $

$\ $

\noindent The main result in this subsection is the following result.

\begin{theorem}\label{thm:almost-surely} 
Let $d: =d_1=d_2$. 
For $Z$ be the random matrix defined as (\ref{eqn:007}), if $W_1, W_2$ are chosen independently, 
we have
$$\frac{1}{d^2} Tr [Z^p] \rightarrow \mathbb{E} \left( \frac{1}{d^2} Tr[Z^p] \right),$$
almost surely as $d \rightarrow \infty, s/d \rightarrow c$. 
\end{theorem}

\begin{proof}
We have 
\begin{equation}\label{eq: almost-surely}
Tr [Z^p] = \sum_{k=0}^p \binom{p}{k} (-1)^{p-k} \left( \frac{s}{d} \right)^{p-k} Tr \left[ \left( \frac{W}{ds} \right)^k \right].
\end{equation}
Hence, it is equivalently to show
$$\frac{1}{d^2} Tr \left[ \left( \frac{W}{ds} \right)^p \right] \rightarrow \mathbb{E} \left( \frac{1}{d^2} Tr \left[ \left( \frac{W}{ds} \right)^p \right]  \right) \;\; \text{almost surely}.$$
By using the Chebyshev inequality and the Borel-Cantelli lemma, it is sufficient 
to show the following inequality:
$$\sum_{d =1}^{\infty}  \left[ \mathbb{E} \left( \frac{1}{d^2}Tr \left[ \left( \frac{W}{ds} \right)^p \right]  \right)^2  - \left( \mathbb{E} \left( \frac{1}{d^2}Tr\left[ \left( \frac{W}{ds} \right)^p \right]  \right) \right)^2  \right] < \infty.$$

When $d=d_1=d_2$ and $s/d=c$, we write (\ref{eq: Moments-indep1}) in Case (I) of Theorem \ref{prop:Moments} as
\begin{equation*}
\frac{1}{d^2} \mathbb{E} \left(Tr\left[ \left( \frac{W}{ds} \right)^p \right] \right) =  \sum_{\pi_1, \pi_2 \in S_p} d^{f(\pi_1, \pi_2)} c^{ -|\pi_1|-|\pi_2| + p}, 
\end{equation*}
where 
\begin{equation*}
\begin{split}
f(\pi_1, \pi_2) & = 2p -|\gamma^{-1} \pi_1| -|\gamma^{-1} \pi_2|-|\pi_1|-|\pi_2|-2+ \#(\pi_1^{-1} \pi_2)-p\\
& = g(\pi_1, \pi_2) + \#(\pi_1^{-1} \pi_2)-p\leq 0, 
\end{split}
\end{equation*}
where 
\begin{equation*}
\begin{split}
g(\pi_1, \pi_2)& = 2p -|\gamma^{-1} \pi_1| -|\gamma^{-1} \pi_2|-|\pi_1|-|\pi_2|-2\\
 &\leq p-|\gamma^{-1} \pi_2|-|\pi_2|-1\leq 0.
\end{split}
\end{equation*}
We hence deduce
that $f(\pi_1, \pi_2) =0$ if and only if $id-\pi_1 = \pi_2 - \gamma.$ 
Note that $|\alpha^{-1} \beta | + | \beta^{-1} \pi|$ has the same parity with $|\alpha^{-1} \pi|$ for any permutation 
$\alpha, \beta, \pi.$ 
Hence the possible values of the function $g(\pi_1, \pi_2)$ are $-2k, k\in \mathbb{N}$. Therefore, we have
\begin{equation*}
\begin{split}
\frac{1}{d^2} \mathbb{E} \left(Tr \left[ \left( \frac{W}{ds} \right)^p \right] \right) & = \sum_{id-\pi_1, \pi_2 -\gamma} c^{ -|\pi_1|-|\pi_2| +p} + O\left( \frac{1}{d^2} \right).
\end{split}
\end{equation*}
Therefore we have 
\begin{equation}\label{eq: almost-surely1}
\begin{split}
\left( \mathbb{E} \left( \frac{1}{d^2}Tr \left[ \left( \frac{W}{ds} \right)^p \right]  \right) \right)^2  & =  \sum_{\begin{subarray}{c}\\ id-\pi_1, \pi_2 -\gamma\\ id-\pi_1', \pi_2' -\gamma \end{subarray}} c^{ -|\pi_1|-|\pi_2|-|\pi_1'| -|\pi_2'|+2p} +   O\left( \frac{1}{d^2} \right).
\end{split}
\end{equation}

The term $\mathbb{E} \left( \frac{1}{d^2}Tr\left[ \left( \frac{W}{ds} \right)^p \right]  \right)^2 $ is more involved to estimate 
and one needs to introduce the permutation $\bar{\gamma} = \gamma \oplus \gamma  \in  S_{2p}.$ By graphical 
Gaussian calculus, we have 
\begin{equation*}
\mathbb{E} \left( \frac{1}{d^2}Tr \left[ \left( \frac{W}{ds} \right)^p \right] \right)^2 = \sum_{\pi_1, \pi_2 \in S_{2p}} d^{\bar{f}(\pi_1, \pi_2)} c^{ -|\pi_1|-|\pi_2| + 2p},
\end{equation*}
where $\bar{f} (\pi_1, \pi_2) = 4p -|\bar{\gamma}^{-1} \pi_1| -|\bar{\gamma}^{-1} \pi_2|-|\pi_1|-|\pi_2|-4 + \#(\pi_1^{-1} \pi_2)-2p.$ This can be done similarly to the proof of Proposition \ref{prop:Moments}, and we leave the details to the reader.

One can easily show that 
\begin{equation*}
\begin{split}
\bar{g} (\pi_1, \pi_2) &:= 4p-|\bar{\gamma}^{-1} \pi_1| -|\bar{\gamma}^{-1} \pi_2|-|\pi_1|-|\pi_2|-4 \\
& \leq 4p-4 + 2 |\bar{\gamma} |\leq  0.
\end{split}
\end{equation*}
The inequality will be saturated when $id- \pi_1, \pi_2 -\bar{\gamma}.$ 
Note that $|\bar{\gamma}^{-1} \pi_1| +|\pi_1|$ and $|\bar{\gamma}^{-1} \pi_2| +|\pi_2|$ has the same parity as $|\bar{\gamma}|=2(p-1)$. Hence, we have
\begin{equation*}
\begin{split}
\mathbb{E} \left( \frac{1}{d^2}Tr \left[ \left( \frac{W}{ds} \right)^p \right]  \right)^2  & =   \sum_{id-\pi_1, \pi_2 -\bar{\gamma}} c^{ -|\pi_1|-|\pi_2| +2p}  + O\left( \frac{1}{d^2} \right).
\end{split}
\end{equation*}
Moreover, if $id-\pi_1, \pi_2 -\bar{\gamma},$ $\pi_1$ and $\pi_2$ can be decomposed into 
$$\pi_1 = \pi_1^{(1)} \oplus \pi_1^{(2)}, $$
$$\pi_2 = \pi_2^{(1)} \oplus \pi_2^{(2)}, $$
where $\pi_i^{(1)} \in S_p$ and $\pi_i^{(2)}\in S_p, i=1,2.$ Finally we have
\begin{equation}\label{eq: almost-surely2}
\begin{split}
\mathbb{E} \left( \frac{1}{d^2}Tr \left[ \left( \frac{W}{ds} \right)^p \right]  \right)^2  & = \sum_{\begin{subarray}{c}\\ id-\pi_1^{(1)}, \pi_1^{(2)} -\gamma\\ id-\pi_2^{(1)}, \pi_2^{(2)}-\gamma \end{subarray}} c^{ -|\pi_1^{(1)}|-|\pi_1^{(2)}|-|\pi_2^{(1)}| -|\pi_2^{(2)}|+2p} +   O\left( \frac{1}{d^2} \right).
\end{split}
\end{equation}
Combining equations \eqref{eq: almost-surely1} and \eqref{eq: almost-surely2}, we have
$$ \mathbb{E} \left( \frac{1}{d^2}Tr \left[ \left( \frac{W}{ds} \right)^p \right]  \right)^2  - \left( \mathbb{E} \left( \frac{1}{d^2}\left[ \left( \frac{W}{ds} \right)^p \right] \right) \right)^2  = O\left( \frac{1}{d^2}\right),$$
which completes our proof. 
\end{proof}

\begin{corollary}
Let $d:=d_1 =d_2,$ if $W_1, W_2$ are chosen independently, 
then the distribution of the random matrix $Z$ defined in (\ref{eqn:007}) converges 
to the centered Marcheko-Pastur distribution of parameter $c^2$ 
(rescalled by the factor $c$)
almost surely as $d \rightarrow \infty, s/d \rightarrow c$. 
\end{corollary}

\begin{remark}
\begin{enumerate}
\item Note that the almost sure convergence of $Z$ to the 
semicircular distribution
also follows from the more general result by 
Bai and Yin \cite{BY1988}, 
which claims the almost sure convergence holds for any fixed $d_1$ under the asymptotic 
regime $d_2 \rightarrow \infty, s/d_2 \rightarrow \infty.$ However, our proof relies on moments techniques, 
and therefore make use of graphical Gaussian calculus, which makes the proof much simpler. 

\item The above results does hold for the case II, because if we let $d_1 =d_2 :=d,$ the power of $d$ terms in 
equation \eqref{eq: almost} is the following:
$2|I| -|\delta_I^{-1} \pi| -|\pi| -2+ \# (\beta_I^{-1}\pi )-|I|  \leq  \# (\beta_I^{-1}\pi )-|I| \leq |I| \neq 0,$ where 
$\pi\in S_{I\sqcup I}$.
Hence by using this moments technique, we do not know whether the convergence in moments holds or not. 

\item The random matrix model used in this section should be compared with the model studied in \cite{CNY2011}. Though these two models are different, their limit distributions are the same. 
\end{enumerate}
\end{remark}


\section{Applications to the PPT square conjecture}

\subsection{PPT square conjecture}
Let us first recall some notations used in quantum information theory. Consider $\mathbb{C}^n=\mathbb{C}^{d_1}\otimes \mathbb{C}^{d_2}$ with dimension $n=d_1d_2$. 
A quantum state $\rho$ on $\mathbb{C}^n$ is a positive operator with $Tr(\rho)=1$.
A state is called \emph{separable} if it can be written as a linear combination of product states. A state $\rho$ that is not separable is cassed \emph{entangled} state. 
A state $\rho$ on $\mathbb{C}^n$ is called PPT (positive partial transpose) if 
$\rho^\Gamma=(Id\otimes T)(\rho)$ is a positive operator, where $T$ is the transpose operator on $M_{d_2}(\mathbb{C})$. By defintion, we see that the partial transpose of
a separable state is always positve. The PPT property is relatively easier to check and is a useful criteria to study entanglement. We refer the reader to \cite{HHHH2007} for some more information about entangled states. 

Let $\rho_1, \rho_2$ be quantum states on $\mathbb{C}^{d_1} \otimes \mathbb{C}^{d_2}$, the typical 
entanglement swapping protocol can be represented as follows: 
\begin{equation}\label{eqn:014}
\rho = \frac{Tr_{d_1} [\rho_1 \otimes \rho_2\; P_{d_1}]}{\text{Normalized factor}},
\end{equation}
where $P_{d_1}$ is the Bell projection on $\mathbb{C}^{d_1} \otimes \mathbb{C}^{d_1}.$

To illustrate the action of entanglement swapping, we look at the case when $d_1 = d_2: =d$ and $\rho_1, \rho_2$ are $d$-dimensional maximally entangled states. Then $\rho$ is also a $d$-dimensional maximally entangled states. This can be easily seen via the graphical language in Figure \ref{fig 3}.

  \begin{figure}[H]
  \centering
  \includegraphics[height=0.25\textwidth]{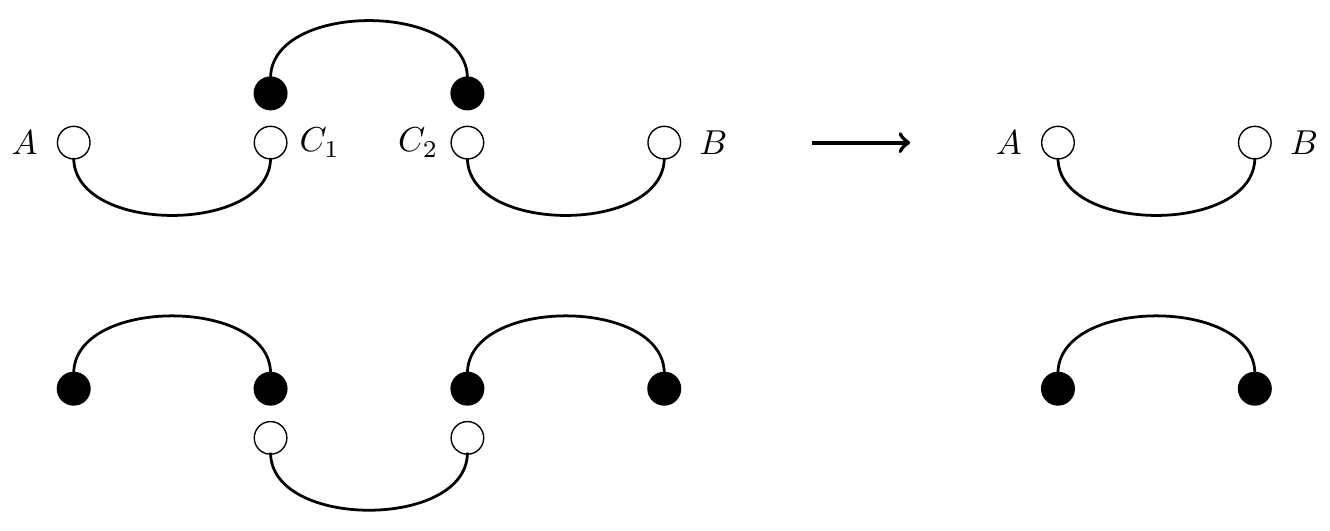}
  \caption{Diagram picture for the entanglement swapping of maximally entangled state.}
  \label{fig 3}
  \end{figure}

The PPT square conjecture suggests that if $\rho_1$ and $\rho_2$ are bound entangled states then $\rho$ defined in (\ref{eqn:014}) is separable. Since the bound entangled state with negative partial transpose is still a mystery, so in practice, we will focus on the PPT bound entangled states \cite{HHH1998}.

\begin{remark}
Due to the Choi-Jamio{\l}kowski isomorphism \cite{J1972, Choi1975} between quantum states and quantum channels, 
there is an equivalent \say{channel} form of the PPT square conjecture 
given by B\"{a}umal \cite[Lemma 14]{B2010} and and Christandl \cite{RJKHW2012}: {\it if $\Phi$ and $\Psi$ are PPT quantum channels,
then their composition $\Phi \circ \Psi$ must be entangled breaking.} Recently, Kennedy-Manor-Paulsen showed that the 
PPT square conjecture holds asymptotically, namely, they proved that the distance between the iteration of any PPT 
channel and the set of all entangled breaking channels approximates to zero \cite{KMP2017}. This result has been improved by Rahaman-Jaques-Paulsen \cite{RJP2018}, where they showed that every unital PPT channel becomes entanglement breaking after a finite number of iterations.  
\end{remark}


\subsection{Random induced state and PPT entanglement threshold}

We denote by $\mu_{n,s}$ the distribution of the induced state $Tr_s |\phi\rangle\langle \phi|,$ where $\phi$ is uniformly
 distributed on the unit sphere in $\mathbb{C}^n \otimes \mathbb{C}^s.$ A random state $\rho_{n,s} $ on 
 $\mathbb{C}^n$ with distribution $\mu_{n, s}$ is called an \emph{induced state}.

The induced states can be described by a random matrix model. 
Let $W = G G^{\ast} \in W(n ,s)$ 
be a Wishart matrix with parameter $(n, s)$. 
It is known that $W/Tr [W]$ is a random state with distribution $\mu_{n,s}$ (see  \cite{N2007} for example).
Moreover, $W/Tr [W]$ and $Tr [W]$ are independent \cite{N2007, Col2005}. 
In addition, $Tr [W]$ is strongly concentrated around $ns$. Hence in this sense we can write
$$\rho_{n, s} = \frac{W} {Tr [W]} \approx \frac{W}{n s} \; \text{for sufficiently large} \; n, s.$$

With the Wishart matrix model, it is possible to estimate the spectrum of $\rho_{n, s}$ in the asymptotic regime 
$s, n \rightarrow \infty$. By using techniques in asymptotic geometric analysis, Aubrun \cite{A2012} and 
Aubrun-Szarek-Ye \cite{ASY2014}  found the following PPT and entanglement threshold of $\rho_{n, s}$ respectively. Note here the partial transpose and separability of $\rho_{n, s}$ are due to the dipartion $\mathbb{C}^n = \mathbb{C}^{d} \otimes
 \mathbb{C}^{d},$ thus $n =d^2.$

\begin{proposition}\label{prop:PPT-shreshold}\cite{A2012}
For given $\epsilon >0,$ 
\begin{enumerate}
\item If $s \leq (4-\epsilon) d^2,$ the probability that $\rho_{d^2, s}$ is PPT is  exponentially decay to 0 as $s\rightarrow \infty.$
\item If $s \geq (4+\epsilon) d^2,$ the probability that $\rho_{d^2, s}$ is PPT is  exponentially decay to 1 as $s\rightarrow \infty.$
\end{enumerate}
\end{proposition}

\begin{proposition}\label{prop:ET-shreshold}\cite{ASY2014}
For given $\epsilon >0,$ there exist constants $C_1, C_2$ and a function $s_0 = s_0(d)$ such that
$$ C_1d^3 \leq s_0 \leq C_2 d^3 \log^2(d) $$
and
\begin{enumerate}
\item  If $s < (1-\epsilon) s_0$, the probability that $\rho_{d^2, s}$ is separable is exponentially decay to $0$ as $s \rightarrow \infty$.
\item  If $s > (1+\epsilon) s_0$, the probability that $\rho_{d^2, s}$ is separable is exponentially decay to $1$ as $s \rightarrow \infty$.
\end{enumerate}
\end{proposition}

Denote by $M_n^{sa, 0}$ the set of all $n \times n$ self-adjoint matrices with trace $0$.  Let $K$ be a convex body in 
$M_{n}^{sa, 0},$ with the $\|\cdot\|_K$ the gauge function defined by  
$\|x\|_K = \inf \{t\geq 0, \; x \in tK \}.$ 
Following \cite{ASY2014},
define a gauge $\phi_K$ on $\mathbb{R}^{n, 0}$ by 
$$\phi_K (x) = \int_{U(n)} \|U Diag (x) U^\ast\|_K dU.$$
 
There are two crucial facts in Aubrun-Szarek-Ye's work \cite{ASY2014}.
\begin{enumerate}
\item $\mathbb{E} \; \phi_K (sp(\rho_{n, s}- \un/n)) = \mathbb{E} \|\rho_{n, s}- \un/n\|_K,$ where $sp(\rho_{n, s}- \un/n)$ is 
the spectrum vector of $(\rho_{n, s}- \un/n)$ in $\mathbb{R}^{n ,0}.$ This is due to the Haar unitary invariance of 
$(\rho_{n, s}- \un/n).$

\item When $n$ and $s/n$ tend to infinity, the empirical spectral distribution of  $\sqrt{ns} (\rho_{n, s}- \un/n)$  converges 
to $\mu_{SC}$, in probability, with respect to the $\infty$-Wasserstein distance. This is because as $n$ and $s/n$ tend to infinity, 
$\sqrt{ns} (\rho_{n, s}- \un/n)$ almost surely converges to $\mu_{SC}.$
\end{enumerate}

With the above two facts,  the gauge of $(\rho_{n, s}- \un/n)$ and $G_n$ can be comparable in the asymptotic regime 
$n\rightarrow \infty, s/n \rightarrow \infty,$ where $G_n$ is a GUE ensemble in $M_n^{sa, 0}$.  More precisely, by \cite[Proposition 3.1]{ASY2014}, we have 
$$\mathbb{E}\;  \left\| \rho_{n,s}- \frac{\un}{n} \right\|_K \approx \mathbb{E}\; \frac{1}{n\sqrt{s}} \left\|G_{n} \right\|_K, \text{as}\; n \rightarrow \infty, s/n \rightarrow \infty.$$

Combining the above two propositions, if we chose the parameter properly such that $4 d^2< s < s_0$, the random state 
$\rho_{n, s}$ would be generically PPT entangled. In other words, $\rho_{n, s}$ is PPT entangled with high probability as 
$n \rightarrow \infty.$

\subsection{PPT square conjecture generically holds when the states are chosed independently}

Let $\rho_1$ and $\rho_2$ are two random induced states with distribution $\mu_{d_1d_2, s}.$ Namely, we can write 
$$\rho_i = \frac{W_i}{Tr [W_i]}, i=1,2,$$
where $W_i \in W(d_1d_2, s), i=1,2.$ Recall the state that is induced by the "entanglement swapping" protocol is the following:
\begin{equation}\label{eq:ES2}
\begin{split}
\rho & = \frac{Tr_{d_1} \left[ \rho_1 \otimes \rho_2 \; P_{d_1}\right]}{\text{Normalized factor} } = \frac{  \frac{1}{d_1} Tr_{d_1} \left[ W_1 \otimes W_2 \; P_{d_1}\right]}{\text{Normalized factor} }\\
& = \frac{W}{Tr [W]},
\end{split}
\end{equation}
where $W$ is defined in (\ref{eqn:004}).

\begin{lemma}\label{lem: final}
If $\rho_1$ and $\rho_2$ are two random induced states with distribution $\mu_{d_1d_2, s},$  then the random state 
$\rho$ on $\mathbb{C}^{d_2} \otimes \mathbb{C}^{d_2}$ converges in moments to $\mu_{d_2^2, s^2}$ as 
$d_1 \rightarrow \infty.$
\end{lemma}

\begin{proof}
Let $W_i = G_i G_i^{\ast},$ where $G_i$'s entries are $\{ g^{(i)}_{jk, t}, j=1, \ldots, d_1, k=1, \ldots, d_2, t=1, \ldots, s\},  i=1, 2.$ After simple calculation we can obtain: 
$$W = G G^{\ast},$$
where  $G$ is a $d_2^2 \times s^2$ matrix with the following entries $\{ g_{kk', tt'}, k, k' =1, \ldots, d_2, t, t' =1, \ldots, s\}$:
\begin{equation}\label{eq: ES-main}
g_{kk', tt'} = \frac{1}{\sqrt{d_1}} \sum_{j=1}^{d_1} g^{(1)}_{jk, t} g^{(2)}_{jk', t'}.
\end{equation}
If we suppose $\rho_1$ and $\rho_2$ are independent, i.e, $g^{(i)}_{jk, t}$ are independent standard normal random 
variables, then due to the classical central limit theorem, the distribution of $g_{kk', tt'}$  converges in moments to a 
standard normal distribution as $d_1 \rightarrow \infty$ for every $k, k', t, t'$. 
Hence $W \in W(d_2^2, s^2)$ when $d_1=\infty$.
\end{proof}

By letting $p=1$ in the equation \eqref{eq: Moments-indep1}, we have 
$\mathbb{E} (Tr [W]) =d_2^2 s^2$. Moreover, $W/Tr [W]$ and $Tr[W]$ are  independent as $d_1\rightarrow\infty$ 
(this is due to Lemma \ref{lem: final} and the discussion in Section 4.2). Moreover, 
$W$ is strongly concentrated around $d_2^2 s^2$ for sufficiently large 
$d_1, d_2$ and $s$. Hence we can formally use $W/(d_2^2s^2)$ to replace
$\rho=W/(Tr[W])$ for sufficiently large $d_1, d_2$ and $s$ (see for \cite[Proposition 6.34]{AS2017} and the discussion after that).
In addition, we can further treat $ W$ following the same distribution as the Wishart matrix $W(d_2^2, s^2)$ for large $d_1$ for our purpose. 

\begin{remark}
The above lemma and discussion provide a direct explanation that the results in Section 3.2 also holds for our model $W$ defined in (\ref{eqn:004}) when $W_1, W_2$ are chosen independently. To make the argument rigrously, we need the almost sure convergence developed in Section 3.3.
\end{remark}

In the rest of this section, we assume $d: =d_1 =d_2,$ and denote $n =d^2.$ Suppose $\rho_1$ and $\rho_2$ are two independent random states with distribution $\mu_{d^2, s}$ which are generically PPT entangled, then the state $\rho$ in equation \eqref{eq:ES2} is generically separable. Our idea is to adapt the arguments in \cite{ASY2014}. To this end, according to the lemma \ref{lem: final} and the discussion, we can approximately write $\rho = W/ns^2$ for sufficiently large $d$ and $s,$ where $W \in W(n ,s^2).$ So instead of $\rho,$ we can consider $W/ns^2$ in this asymptotic picture ($d, s\rightarrow \infty$). Recall there are two important ingredients in \cite{ASY2014}. The fact that $\rho = W/n s^2$ enables us to  
compare the gauge of the spectrum vector of $\rho- \un/n$ and itself, this is due to the Haar unitary invariance of $W.$ On the other hand, for the second ingredient, we have to use our theorem \ref{thm:almost-surely}, by which we can compare the gauge $\rho-\un/n$ and the GUE ensemble $G_n.$ Then by using the concentration of measure technique (see for instance  \cite[Section 2.2]{ASY2014}), we can show the separability of $\rho.$ 

In conclusion, we can roughly say that the distribution of $\rho$ is $\mu_{n , s^2},$ and the parameter $s^2$ makes $\rho$ is generically separable. The following is our final theorem.    

\begin{theorem}
Let $\rho_1$ and $\rho_2$ are two independent random states with distribution $\mu_{d^2, s}$ which are generically 
PPT entangled, then the state $\rho$ in equation \eqref{eq:ES2} is generically separable.
\end{theorem}

\begin{proof}

By Lemma \ref{lem: final} and the discussion above, 
we can formally write 
$$\rho = \frac{W}{n s^2}, \; as  \; d, s \rightarrow \infty,$$
where $W \in W(d^2, s^2).$  Denote  by $\mathcal{S}$ the set of all separable quantum states on $\mathbb{C}^d \otimes \mathbb{C}^d$, 
and let $\mathcal{S}_0 = \mathcal{S}- \un /n.$ Obviously, $\mathcal{S}_0$ is a convex set of $M_{n}^{sa, 0}.$ Hence for convex body $\mathcal{S}_0$ in $M^{sa, 0}_{n},$ we have 
$$ \mathbb{E} \; \phi_{\mathcal{S}_0} (sp(\rho- \un/n)) =   \mathbb{E}\; \left\|\rho-\frac{\un}{n} \right\|_{\mathcal{S}_0}, \text{as} \; d \rightarrow \infty.$$

Here let us mention that in \cite{ASY2014} the above equation holds for arbitrary $d.$ However, the condition $d \rightarrow \infty$ is necessary in our paper, since $\rho$ is Haar unitary invariance only if $d= d_1 \rightarrow \infty,$ which is a necessary condition for the equation. 
On the other hand, combining with the theorem \ref{thm:almost-surely},  $d s (\rho- \un/n)$ almost surely 
converges to $\mu_{SC}$ as $d \rightarrow \infty, s/d \rightarrow \infty.$ Similar to \cite[Proposition 3.1]{ASY2014} we have
\[
\mathbb{E}\;  \left\| \rho- \frac{\un}{n} \right\|_{\mathcal{S}_0} \approx \mathbb{E}\; \frac{1}{n s} \left\|G_{n} \right\|_{\mathcal{S}_0}, \text{as} \; d \rightarrow \infty, s/d \rightarrow \infty,
\]
where $G_{n}$ is a GUE ensemble in $M_{n}^{sa, 0}.$  The symbol "$\approx$" means the terms of left (right) hand side are bounded by each other (up to constants $c_{d, s}, C_{d, s}$, and $ c_{d, s}, C_{d, s} \rightarrow1 as d \rightarrow \infty, s/d \rightarrow \infty$).

Recall that the quantum 
state $\rho$ on $\mathbb{C}^d \otimes \mathbb{C}^d$ is separable if and only if \cite{ASY2014}

$$\left\| \rho \right\|_{\mathcal{S}_0} \leq 1 \;  \Longleftrightarrow  \;\left\| \rho- \frac{\un}{n} \right\|_{\mathcal{S}_0} \leq 1.$$
Moreover, we have, by \cite[Section 4]{ASY2014},
$$\mathbb{E} \left\|\frac{G_n}{n} \right\|_{\mathcal{S}_0} = O(d^{\frac{3}{2}} \log d).$$
Hence 
\begin{equation*}
\mathbb{E}\; \left \|\rho- \frac{\un}{n} \right\|_{\mathcal{S}_0} = \frac{O(d^{\frac{3}{2}} \log d)}{s}.
\end{equation*}

By the concentration of measure technique, for any $t>0$ we have
\begin{equation*}
\mathbb{P} \left(\left \|\rho- \frac{\un}{n} \right\|_{\mathcal{S}_0}  > \mathbb{E}\; \left \|\rho- \frac{\un}{n} \right\|_{\mathcal{S}_0} + t \right) \leq e^{-s^2t^2/n},
\end{equation*}
where we have used the fact that $\rho \rightarrow  \|\rho\|_{\mathcal{S}_0}$ is a $2n$-Lipschitz function 
(see \cite[Lemma 3.4]{ASY2014}) on the real 
sphere $S^{2ns^2-1}.$ If $s^2 >  s_0= d^3 \log^2 d$, then
\begin{equation*}
\mathbb{P} \left(\left \|\rho- \frac{\un}{n} \right\|_{\mathcal{S}_0}  > 1 + t \right) \leq e^{-s^2t^2/n} \leq e^{-d t^2}.
\end{equation*}
Hence the probability that $\rho$ is entangled decays exponentially to $0$ as $d \rightarrow \infty.$ However, since 
$\rho_1$ and $\rho_2$ are PPT entangled, the required parameters should satisfy $s > 4d^2.$ Thus $s^2 > 16d^4 > s_0$ 
for sufficiently large $d,$ which implies $\rho$ is generically separable. 
\end{proof}

\begin{remark}
The above discussion on concentration of measure techniques
is an adaption of S. J. Szarek's lecture \say{Geometric Functional Analysis 
and QIT} \cite{S2017} in the trimester program of the Centre Emile Borel
\say{Analysis in Quantum Information Theory}. \end{remark}


\section{Conclusion}

In this paper, we have studied the random matrix which is obtained by \say{entanglement swapping} protocol of two 
Wishart matrices. By using Gaussian graphical calculus, we are able to obtain some limit theorems of our random 
models, namely, they converge to the Marcenko-Pastur law (resp. semi-circle law) under proper asymptotic regime. 
An interesting application is that we have proved the PPT square conjecture holds generically if we independently chose 
the states.  

Although we hoped to obtain counterexamples and therefore settle the problem, our investigations so far 
tend rather to serve as evidence that the PPT square conjecture might be true (at least, very often, in a natural probabilistic sense). 
We considered many random models, including non-independent models, and but they 
seem to yield similar conclusions, so we focus on a specific natural model where both maps are independent. 

However, we still believe that there might exist counterexample for this conjecture in high dimension. Obviously, to obtain 
the counterexample, there must exist correlation between the chosen states. To this end, we have considered the 
random model obtained by the same states. But unfortunately, it does not strongly converges to a required probability
 distribution, and even worse, the distribution of the model is not unitary invariant. Hence it is not possible to use the 
 existence techniques in the asymptotic geometric analysis.

\vspace{3mm}

 {\it Acknowledgments:} B.C. is supported by JSPS Kakenhi 17K18734, 
 15KK0162 and  17H04823,  and 
 ANR-14-CE25-0003. Z.Y. is supported by NSFC No. 11771106 and No. 11431011.
P.Z. is partially supported by NSFC (No. 11501423, No.
11671308) and NSERC research grant.
P.Z. also wants to thank IASM of the Harbin Institute of Technology for warm hospitality in summer 2016.
The authors are grateful to the IHP for holding the trimester 
\say{Analysis in Quantum Information Theory}
during the fall 2017, that provided a fruitful environment to meet and work on the project of this paper.


\bibliographystyle{plain}

\begin{thebibliography}{1}
\bibitem{MM2017} It was claimed by Muller Mermes in the open problem section of the trimester program of the Centre Emile Borel
\say{Analysis in Quantum Information Theory}.
\end{thebibliography}

\end{document}